\documentclass[10pt]{article}

\usepackage{graphics}
\usepackage{pst-all}
\usepackage{color}
\usepackage{pstricks} 
\usepackage{pst-optic}
\usepackage{pst-text}
\usepackage[cp1250]{inputenc}
\usepackage{enumerate}
\usepackage{amsmath,amssymb}

\newcommand{\cc}{\mbox{$\cal C$} }
\newcommand{\cd}{\mbox{$\cal D$} }
\newcommand{\co}{\mbox{$\cal O$} }
\newcommand{\cp}{\mbox{$\cal P$} }
\newcommand{\cs}{\mbox{$\cal S$} }
\newcommand{\kw}{\rule{2mm}{2mm}}
\newcommand{\qedd}{\hfill \kw}

\newtheorem{pro}{Proposition}
\newtheorem{obs}{Observation}
\newtheorem{thm}{Theorem}
\newtheorem{lem}{Lemma}
\newtheorem{cor}{Corollary}

\newtheorem{op}{Open Problem}
 
\newtheorem{claim1}{Claim}

\newcommand{\sm}{\smallskip}
\newcommand{\ms}{\medskip}

\newcommand{\noi}{\noindent}

\newcommand{\tabu}{\hspace{15pt}}

\newenvironment{proof}{\noindent\textbf{Proof.}}{\hfill$\qedd$\medskip} 

\newcommand{\grafGR}{
  \psset{unit=2.0cm}  
  \begin{pspicture}(-0.2,-0.2)(5.20,3.80)
    \pnode(3.70,3.10){v1}
    \pnode(2.50,2.30){v2}
    \pnode(3.10,2.30){v3}
    \pnode(4.30,2.30){v4}
    \pnode(4.90,2.30){v5}
    \pnode(1.30,1.50){v6}
    \pnode(1.90,1.50){v7}
    \pnode(3.10,1.50){v8}
    \pnode(3.70,1.50){v9}
    \pnode(0.10,0.70){v10}
    \pnode(0.70,0.70){v11}
    \pnode(1.90,0.70){v12}
    \pnode(2.50,0.70){v13}
    \pnode(3.70,2.30){v14}
    \pnode(0.10,0.10){v15}
    \pnode(0.70,0.10){v16}
    \pnode(1.90,0.10){v17}
    \psline[linecolor=black,linewidth=0.2pt](v2)(v1)
    \psline[linecolor=black,linewidth=0.2pt](v3)(v1)
    \psline[linecolor=black,linewidth=0.2pt](v4)(v1)
    \psline[linecolor=black,linewidth=0.2pt](v5)(v1)
    \psline[linecolor=black,linewidth=0.2pt](v6)(v2)
    \psline[linecolor=black,linewidth=0.2pt](v7)(v2)
    \psline[linecolor=black,linewidth=0.2pt](v8)(v2)
    \psline[linecolor=black,linewidth=0.2pt](v9)(v2)
    \psline[linecolor=black,linewidth=0.2pt](v10)(v6)
    \psline[linecolor=black,linewidth=0.2pt](v11)(v6)
    \psline[linecolor=black,linewidth=0.2pt](v12)(v6)
    \psline[linecolor=black,linewidth=0.2pt](v13)(v6)
    \psline[linecolor=black,linewidth=0.2pt](v14)(v1)
    \psline[linecolor=black,linewidth=0.2pt](v15)(v10)
    \psline[linecolor=black,linewidth=0.2pt](v16)(v11)
    \psline[linecolor=black,linewidth=0.2pt](v17)(v12)
    \pscircle[fillcolor=white, fillstyle=solid,linewidth=0.1pt](v1){0.09}
      \uput[34.29](v1){${v}$}
     \rput(v1){${1}$}
    \pscircle[fillcolor=white, fillstyle=solid,linewidth=0.1pt](v2){0.09}
      \uput[34.99](v2){${x}$}
     \rput(v2){${1}$}
    \pscircle[fillcolor=white, fillstyle=solid,linewidth=0.1pt](v3){0.09}
      \uput[23.63](v3){${y}$}
      \rput(v3){${1}$}
    \pscircle[fillcolor=white, fillstyle=solid,linewidth=0.1pt](v4){0.09}
      \uput[14.04](v4){${t}$}
            \rput(v4){${1}$}
    \pscircle[fillcolor=white, fillstyle=solid,linewidth=0.1pt](v5){0.09}
      \uput[11.63](v5){${w}$}
     \rput(v5){${1}$}
    \pscircle[fillcolor=white, fillstyle=solid,linewidth=0.1pt](v6){0.09}
      \uput[0](v6){${x_1}$}
      \rput(v6){${2}$}
    \pscircle[fillcolor=white, fillstyle=solid,linewidth=0.1pt](v7){0.09}
      \uput[0](v7){${x_2}$}
      \rput(v7){${3}$}
    \pscircle[fillcolor=white, fillstyle=solid,linewidth=0.1pt](v8){0.09}
      \uput[0](v8){${x_3}$}
    \rput(v8){${2}$}
    \pscircle[fillcolor=white, fillstyle=solid,linewidth=0.1pt](v9){0.09}
      \uput[0](v9){${x_4}$}
     \rput(v9){${3}$}
    \pscircle[fillcolor=white, fillstyle=solid,linewidth=0.1pt](v10){0.09}
      \uput[0](v10){${x'_1}$}
      \rput(v10){${1}$}
    \pscircle[fillcolor=white, fillstyle=solid,linewidth=0.1pt](v11){0.09}
      \uput[0](v11){${x'_2}$}
      \rput(v11){${2}$}
    \pscircle[fillcolor=white, fillstyle=solid,linewidth=0.1pt](v12){0.09}
      \uput[-0](v12){${x'_3}$}
      \rput(v12){${4}$}
    \pscircle[fillcolor=white, fillstyle=solid,linewidth=0.1pt](v13){0.09}
     \uput[0](v13){${x'_4}$}
      \rput(v13){${\alpha}$}
    \pscircle[fillcolor=white, fillstyle=solid,linewidth=0.1pt](v14){0.09}
      \uput[17.65](v14){${z}$}
            \rput(v14){${1}$}
    \pscircle[fillcolor=white, fillstyle=solid,linewidth=0.1pt](v15){0.09}
      \uput[0](v15){${p}$}
     \rput(v15){${2}$}
    \pscircle[fillcolor=white, fillstyle=solid,linewidth=0.1pt](v16){0.09}
      \uput[0](v16){${q}$}
      \rput(v16){${2}$}
    \pscircle[fillcolor=white, fillstyle=solid,linewidth=0.1pt](v17){0.09}
      \uput[0](v17){${r}$}
      \rput(v17){${2}$}
  \end{pspicture}
}\newcommand{\grafGS}{
  \psset{unit=2.0cm}  
  \begin{pspicture}(-0.2,-0.2)(5.20,3.80)
    \pnode(3.70,3.10){v1}
    \pnode(2.50,2.30){v2}
    \pnode(3.10,2.30){v3}
    \pnode(4.30,2.30){v4}
    \pnode(4.90,2.30){v5}
    \pnode(1.30,1.50){v6}
    \pnode(1.90,1.50){v7}
    \pnode(3.10,1.50){v8}
    \pnode(3.70,1.50){v9}
    \pnode(0.10,0.70){v10}
    \pnode(0.70,0.70){v11}
    \pnode(1.90,0.70){v12}
    \pnode(2.50,0.70){v13}
    \pnode(3.70,2.30){v14}
    \pnode(0.10,0.10){v15}
    \pnode(0.70,0.10){v16}
    \pnode(1.90,0.10){v17}
    \psline[linecolor=black,linewidth=0.2pt](v2)(v1)
    \psline[linecolor=black,linewidth=0.2pt](v3)(v1)
    \psline[linecolor=black,linewidth=0.2pt](v4)(v1)
    \psline[linecolor=black,linewidth=0.2pt](v5)(v1)
    \psline[linecolor=black,linewidth=0.2pt](v6)(v2)
    \psline[linecolor=black,linewidth=0.2pt](v7)(v2)
    \psline[linecolor=black,linewidth=0.2pt](v8)(v2)
    \psline[linecolor=black,linewidth=0.2pt](v9)(v2)
    \psline[linecolor=black,linewidth=0.2pt](v10)(v6)
    \psline[linecolor=black,linewidth=0.2pt](v11)(v6)
    \psline[linecolor=black,linewidth=0.2pt](v12)(v6)
    \psline[linecolor=black,linewidth=0.2pt](v13)(v6)
    \psline[linecolor=black,linewidth=0.2pt](v14)(v1)
    \psline[linecolor=black,linewidth=0.2pt](v15)(v10)
    \psline[linecolor=black,linewidth=0.2pt](v16)(v11)
    \psline[linecolor=black,linewidth=0.2pt](v17)(v12)
    \pscircle[fillcolor=white, fillstyle=solid,linewidth=0.1pt](v1){0.09}
      \uput[34.29](v1){${v}$}
     \rput(v1){${1}$}
    \pscircle[fillcolor=white, fillstyle=solid,linewidth=0.1pt](v2){0.09}
      \uput[34.99](v2){${x}$}
     \rput(v2){${1}$}
    \pscircle[fillcolor=white, fillstyle=solid,linewidth=0.1pt](v3){0.09}
      \uput[23.63](v3){${y}$}
      \rput(v3){${1}$}
    \pscircle[fillcolor=white, fillstyle=solid,linewidth=0.1pt](v4){0.09}
      \uput[14.04](v4){${t}$}
            \rput(v4){${1}$}
    \pscircle[fillcolor=white, fillstyle=solid,linewidth=0.1pt](v5){0.09}
      \uput[11.63](v5){${w}$}
     \rput(v5){${1}$}
    \pscircle[fillcolor=white, fillstyle=solid,linewidth=0.1pt](v6){0.09}
      \uput[0](v6){${x_1}$}
      \rput(v6){${2}$}
    \pscircle[fillcolor=white, fillstyle=solid,linewidth=0.1pt](v7){0.09}
      \uput[0](v7){${x_2}$}
      \rput(v7){${3}$}
    \pscircle[fillcolor=white, fillstyle=solid,linewidth=0.1pt](v8){0.09}
      \uput[0](v8){${x_3}$}
    \rput(v8){${2}$}
    \pscircle[fillcolor=white, fillstyle=solid,linewidth=0.1pt](v9){0.09}
      \uput[0](v9){${x_4}$}
     \rput(v9){${3}$}
    \pscircle[fillcolor=white, fillstyle=solid,linewidth=0.1pt](v10){0.09}
      \uput[0](v10){${x'_1}$}
      \rput(v10){${1}$}
    \pscircle[fillcolor=white, fillstyle=solid,linewidth=0.1pt](v11){0.09}
      \uput[0](v11){${x'_2}$}
      \rput(v11){${4}$}
    \pscircle[fillcolor=white, fillstyle=solid,linewidth=0.1pt](v12){0.09}
      \uput[0](v12){${x'_3}$}
      \rput(v12){${5}$}
    \pscircle[fillcolor=white, fillstyle=solid,linewidth=0.1pt](v13){0.09}
     \uput[0](v13){${x'_4}$}
      \rput(v13){${\alpha}$}
    \pscircle[fillcolor=white, fillstyle=solid,linewidth=0.1pt](v14){0.09}
      \uput[17.65](v14){${z}$}
            \rput(v14){${1}$}
    \pscircle[fillcolor=white, fillstyle=solid,linewidth=0.1pt](v15){0.09}
      \uput[0](v15){${p}$}
     \rput(v15){${2}$}
    \pscircle[fillcolor=white, fillstyle=solid,linewidth=0.1pt](v16){0.09}
      \uput[0](v16){${q}$}
      \rput(v16){${2}$}
    \pscircle[fillcolor=white, fillstyle=solid,linewidth=0.1pt](v17){0.09}
      \uput[0](v17){${r}$}
      \rput(v17){${2}$}
  \end{pspicture}
}
\newcommand{\grafGT}{
  \psset{unit=2.0cm}  
  \begin{pspicture}(-0.2,-0.2)(4.00,3.80)
    \pnode(2.50,3.10){v1}
    \pnode(1.30,2.30){v2}
    \pnode(1.90,2.30){v3}
    \pnode(3.10,2.30){v4}
    \pnode(3.70,2.30){v5}
    \pnode(0.10,1.50){v6}
    \pnode(0.70,1.50){v7}
    \pnode(1.90,1.50){v8}
    \pnode(2.50,1.50){v9}
    \pnode(2.50,2.30){v10}
    \pnode(1.30,0.70){v11}
    \pnode(1.90,0.70){v12}
    \pnode(3.10,0.70){v13}
    \pnode(3.70,0.70){v14}
    \pnode(1.30,0.10){v15}
    \pnode(1.90,0.10){v16}
    \pnode(3.10,0.10){v17}
    \psline[linecolor=black,linewidth=0.2pt](v2)(v1)
    \psline[linecolor=black,linewidth=0.2pt](v3)(v1)
    \psline[linecolor=black,linewidth=0.2pt](v4)(v1)
    \psline[linecolor=black,linewidth=0.2pt](v5)(v1)
    \psline[linecolor=black,linewidth=0.2pt](v6)(v2)
    \psline[linecolor=black,linewidth=0.2pt](v7)(v2)
    \psline[linecolor=black,linewidth=0.2pt](v8)(v2)
    \psline[linecolor=black,linewidth=0.2pt](v9)(v2)
    \psline[linecolor=black,linewidth=0.2pt](v10)(v1)
    \psline[linecolor=black,linewidth=0.2pt](v11)(v9)
    \psline[linecolor=black,linewidth=0.2pt](v12)(v9)
    \psline[linecolor=black,linewidth=0.2pt](v13)(v9)
    \psline[linecolor=black,linewidth=0.2pt](v14)(v9)
    \psline[linecolor=black,linewidth=0.2pt](v15)(v11)
    \psline[linecolor=black,linewidth=0.2pt](v16)(v12)
    \psline[linecolor=black,linewidth=0.2pt](v17)(v13)
    \pscircle[fillcolor=white, fillstyle=solid,linewidth=0.1pt](v1){0.09}
      \uput[28.18](v1){${v}$}
     \rput(v1){${1}$}
    \pscircle[fillcolor=white, fillstyle=solid,linewidth=0.1pt](v2){0.09}
      \uput[23.63](v2){${x}$}
     \rput(v2){${1}$}
    \pscircle[fillcolor=white, fillstyle=solid,linewidth=0.1pt](v3){0.09}
      \uput[17.65](v3){${y}$}
     \rput(v3){${1}$}
    \pscircle[fillcolor=white, fillstyle=solid,linewidth=0.1pt](v4){0.09}
      \uput[11.63](v4){${t}$}
      \rput(v4){${1}$}
    \pscircle[fillcolor=white, fillstyle=solid,linewidth=0.1pt](v5){0.09}
      \uput[9.93](v5){${w}$}
      \rput(v5){${1}$}
    \pscircle[fillcolor=white, fillstyle=solid,linewidth=0.1pt](v6){0.09}
      \uput[-14.04](v6){${x_1}$}
      \rput(v6){${2}$}
    \pscircle[fillcolor=white, fillstyle=solid,linewidth=0.1pt](v7){0.09}
      \uput[-5.71](v7){${x_2}$}
      \rput(v7){${3}$}
    \pscircle[fillcolor=white, fillstyle=solid,linewidth=0.1pt](v8){0.09}
      \uput[-2.60](v8){${x_3}$}
      \rput(v8){${3}$}
    \pscircle[fillcolor=white, fillstyle=solid,linewidth=0.1pt](v9){0.09}
      \uput[-2.05](v9){${x_4}$}
      \rput(v9){${3}$}
    \pscircle[fillcolor=white, fillstyle=solid,linewidth=0.1pt](v10){0.09}
      \uput[14.04](v10){${z}$}
      \rput(v10){${1}$}
    \pscircle[fillcolor=white, fillstyle=solid,linewidth=0.1pt](v11){0.09}
     \uput[0](v11){${x'_1}$}
      \rput(v11){${2}$}
    \pscircle[fillcolor=white, fillstyle=solid,linewidth=0.1pt](v12){0.09}
      \uput[0](v12){${x'_2}$}
     \rput(v12){${3}$}
    \pscircle[fillcolor=white, fillstyle=solid,linewidth=0.1pt](v13){0.09}
      \uput[0](v13){${x'_3}$}
      \rput(v13){${4}$}
    \pscircle[fillcolor=white, fillstyle=solid,linewidth=0.1pt](v14){0.09}
      \uput[0](v14){${x'_4}$}
      \rput(v14){${\alpha}$}
    \pscircle[fillcolor=white, fillstyle=solid,linewidth=0.1pt](v15){0.09}
      \uput[0](v15){${p}$}
      \rput(v15){${3}$}
    \pscircle[fillcolor=white, fillstyle=solid,linewidth=0.1pt](v16){0.09}
      \uput[0](v16){${q}$}
      \rput(v16){${3}$}
    \pscircle[fillcolor=white, fillstyle=solid,linewidth=0.1pt](v17){0.09}
      \uput[0](v17){${r}$}
      \rput(v17){${3}$}
  \end{pspicture}
}
\begin{document}

\begin{center}

	{\bf \Large Acyclic  colourings of graphs with  bounded degree}
	
	\ms
	\ms
	
	{ Anna Fiedorowicz, El\.zbieta Sidorowicz}\\

		\ms
		\ms
		{\it Faculty of Mathematics, Computer Science and Econometrics, \\
			University of Zielona G{\'o}ra,\\
			Z.~Szafrana 4a, 65-516 Zielona G{\'o}ra,
			Poland
		}
		
		\ms
		\ms
{\small e-mail: \it A.Fiedorowicz@wmie.uz.zgora.pl}\\		
		{\small e-mail: \it E.Sidorowicz@wmie.uz.zgora.pl}\\

\end{center}

\begin{abstract}
A $k$-colouring (not necessarily proper) of vertices of a graph is called {\it acyclic}, if for every pair of distinct colours $i$ and $j$ the subgraph induced by the edges whose endpoints have colours $i$ and $j$ is acyclic. In the paper we consider some generalised acyclic $k$-colourings, namely, we require that each colour class induces an acyclic or bounded degree graph.  
Mainly we focus on graphs with maximum degree 5. We prove that any such graph has an acyclic $5$-colouring such that each colour class induces an acyclic graph with maximum degree  at most 4. We  prove that the problem of deciding whether a graph $G$ has an acyclic 2-colouring in which each colour class induces a graph with maximum degree at most 3 is NP-complete, even for graphs with maximum degree 5. We also give a linear-time algorithm for an acyclic $t$-improper  colouring of  any graph  with maximum degree $d$ assuming that the number  of colors is large enough.

\noi {\bf Keywords:} {acyclic colouring,  bounded degree graph}

	\noi {\bf 2010 Mathematics Subject Classification:} 05C15
 \end{abstract}

\section{Introduction}

We consider only finite, simple graphs. We use standard notation. For a graph $G$, we denote its vertex set and edge set by $V(G)$ and $E(G)$, respectively. Let $v\in V(G)$. By $N_G(v)$ (or $N(v)$) we denote the set of the neighbours of $v$ in $G$. The cardinality of $N_G(v)$ is called the {\it degree} of $v$, denoted by $d_G(v)$ (or $d(v)$). The maximum and minimum vertex degrees in $G$ are denoted by $\Delta(G)$ and $\delta(G)$, respectively. For undefined concepts, we refer the reader to \cite{we01}.

\medskip
A \emph{$k$-colouring} of a graph $G$ is a mapping $c$ from the set of vertices of $G$ to the set  \begin{math}\{1,2,\ldots,k\}\end{math} of colours. We can also regard a $k$-colouring of $G$ as a partition of the set $V(G)$ into \emph{colour classes} \begin{math}V_1,V_2,\ldots,V_k\end{math}   such that each $V_i$ is the set of vertices with colour $i$.  In many situations it is desired that the subgraph induced by each set $V_i$ has a given property. For instance, requiring that each set $V_i$ is independent defines a proper $k$-colouring. If each set $V_i$ induces a graph with a given property we obtain a generalised colouring. One can also require that for any pair of distinct colours $i$ and $j$, the subgraph induced by the edges whose endpoints have colours $i$ and $j$ satisfies a given property, for instance, is acyclic. This yields to the concept of acyclic colouring. In this paper we  concentrate on such generalised acyclic colourings, in which each colour class induces an  acyclic graph or a graph with bounded degree.  Below we state precisely this notion. The terminology and notation concerning generalised colourings, which we follow here, can be found in \cite{bobr97}, while the concept of generalised acyclic colourings was introduced in \cite{boso99}. 
 
Let \begin{math}{\cal P}_1,{\cal P}_2,\ldots, {\cal P}_k\end{math} be nonempty classes of graphs closed with respect to isomorphism. 
A $k$-colouring  of a graph $G$ is called a {\begin{math}({\cal P}_1,{\cal P}_2,\ldots, {\cal P}_k)\end{math}-\emph {colouring}}  of  $G$ if for each  $i\in\{1,2,\ldots,k\}$ the subgraph induced in $G$  by the colour class $V_i$ belongs to ${\cal P}_i$. Such a colouring is called an \emph {acyclic} {\begin{math}({\cal P}_1,{\cal P}_2,\ldots,$ $ {\cal P}_k)\end{math}-\emph{colouring}} if  for every two distinct colours $i$ and $j$ the subgraph induced by the edges whose endpoints have colours $i$ and $j$ is acyclic. 
In other words,  every bichromatic cycle  in $G$ contains at least one monochromatic edge.      
Throughout this paper we use the following notation: ${\cal S}_d$ for the class of graphs with maximum degree at most $d$, and ${\cal D}_1$ for the class of acyclic graphs. For convenience, an acyclic $(\cp_1,\cp_2,\ldots,\cp_k)$-colouring, where $\cp_i=\cp$, for each $i\in\{1,2,\ldots,k\}$,  is referred to as an acyclic $\cp^{(k)}$-colouring.\\
\begin{tabular}{l}
\hline
	\tiny{Paper accepted for publication in Science China Mathematics}
\end{tabular}

With this notation, an acyclic $k$-colouring of a graph $G$ corresponds to an acyclic  $\cp^{(k)}$-colouring of $G$, where the class ${\cal P}$ is the set of all edgeless graphs. We use $\chi _a(G)$ to denote the  acyclic  chromatic number.  An acyclic   $\cp^{(k)}$-colouring of $G$ such that ${\cal P}={\cal S}_d$ is called an  \emph{acyclic $d$-improper $k$-colouring}. 

The notion of acyclic colouring of graphs was introduced in 1973 by Gr\H{u}nbaum \cite{gr73} and has been widely considered in the recent past. Even more attention has been payed to this problem since it was proved by  Coleman et al. \cite{coca86,como84} that acyclic colorings can be used in computing Hessian matrices via the substitution method, see also \cite{gema05}.

However, determining $\chi_a(G)$ is quite difficult.   Kostochka \cite{ko78} proved that it is an NP-complete problem to decide for a given arbitrary graph $G$ whether $\chi_a(G)\le 3$.  The acyclic chromatic number was determined for several families of graphs. In particular, this parameter was studied intensively for the family of graphs with fixed maximum degree. Using the probabilistic method, Alon et al. \cite{almc91} showed that any graph of maximum degree $\Delta $ can be acyclically coloured using $O(\Delta ^{4/3})$ colours. A $t$-improper analogues of this result was obtained by Addario-Berry et al. \cite{ades10}. They proved that any graph of maximum degree $\Delta $ has an acyclic $t$-improper colouring with $O(\Delta {\rm ln}\Delta +(\Delta-t)\Delta)$ colours.

Focusing on the family of graphs with  small maximum degree, it was shown in \cite{gr73} that $\chi_a(G)\le 4$ for any graph with maximum degree 3 (see also \cite{sa04}).  Burstein \cite{bu79} proved that $\chi_a(G)\le 5$ for any graph with maximum degree 4. Recently, Kostochka and Stocker \cite{kost11} proved that $\chi_a(G)\le 7$ for any graph with maximum degree 5. For graphs with maximum degree 6, Hocquard \cite{ho11} proved that 11 colours are enough for an acyclic colouring.

In 1999, Boiron et al. began with the study on the problem of acyclic \begin{math}({\cal P}_1,{\cal P}_2,\ldots,$ $ {\cal P}_k)\end{math}-colourings of outerplanar and planar graphs \cite{boso99a}, and bounded degree graphs \cite{boso99}. In particular, they proved that any graph $G\in {\cal S}_3$ has an acyclic $({\cal D}_1,{\cal S}_2)$-colouring as well as an acyclic ${\cal S}_1^{(3)}$-colouring \cite{boso99}. Addario-Berry et al. \cite{adka10} proved that each graph from ${\cal S}_3$ has an acyclic $({\cal S}_2,{\cal S}_2)$-colouring. This theorem was also proved in \cite{bofi10}, where a polynomial-time algorithm was presented. In \cite{boje11}  a polynomial-time algorithm that provides an acyclic $({\cal D}_1,{\cal LF})$-colouring of any graph from ${\cal S}_3\setminus \{K_4,K_{3,3}\}$ was given (${\cal LF}$ is the set of acyclic graphs with maximum degree at most 2). Related problems concerning the class of graphs with maximum degree at most 4 are considered in  \cite{fisi13}, where it was proved that any graph from $\cs_4$ has an acyclic $\cs_3^{(3)}$-colouring, as well as an acyclic $(\cs_3\cap \cd_1)^{(4)}$-colouring. In the present paper we continue the previous work and consider  acyclic colourings of graphs with  maximum degree at most 5.  We prove that each graph $G\in{\cal S}_5$ has an acyclic $\cd_1^{(5)}$-colouring. The number of colours in this theorem cannot be reduced, since the complete graph $K_6$ needs at least 5 colours in any such colouring. Next, we slightly improve this result, by proving  that any graph from ${\cal S}_5$ can be acyclically coloured with $5$ colours in such a way that each colour class induces an acyclic graph with maximum degree at most $4$. We also show that the problem of deciding whether a given graph $G\in {\cal S}_5$ is acyclically $({\cal S}_3,{\cal S}_3)$-colourable is NP-complete. We finish the paper with a general result   giving linear-time algorithms for acyclic $t$-improper colourings of graphs with maximum degree $d$ assuming that the number of colors is large enough with respect to $d$.

\ms

The following definitions and notation, which will be used later in the proofs, deal with a  {\it partial $k$-colouring} of a graph $G$, defined as an assignment $c$ of colours from the set $\{1,2,\ldots,k\}$ to a subset $C$ of $V(G)$.  Given a partial $k$-colouring $c$ of $G$, the set $C$ is the set of coloured vertices.
For a vertex $v$ let  $n_v=|C\cap N(v)|$ and $p_v=|\bigcup_{u\in N(v)\cap C} c(u)|$. Clearly, $p_v\le n_v$. Let $C_v$ denote the multiset of colours assigned by $c$ to the coloured neighbors of $v$. For $S\subseteq V(G)$ let $C(S)=\bigcup _{v\in S}C_v$.
A vertex $v$ is called {\it rainbow}, if all its coloured neighbours  have distinct colours. 


Let $c$ be a partial $k$-colouring of $G$ and $i,j$ be distinct colours. A bichromatic cycle (resp. path) having no monochromatic edge is called an \textit{alternating cycle} (resp. \textit{path}). An $(i,j)$-alternating cycle (path) is an alternating cycle (path) with each vertex coloured $i$ or $j$. Let $F$ be a cycle in $G$ containing $v$. Then $F$ is called $(i,j)$-{\it dangerous for $v$},  if colouring $v$ with $i$ results in  an $(i,j)$-alternating cycle. $F$ is called {\it $i$-mono-dangerous} for $v$, if colouring $v$ with $i$ results in a monochromatic cycle containing $v$. When it is convenient, all $(i,j)$-dangerous cycles and $k$-mono-dangerous cycles for $v$ will be called simply {\it dangerous} cycles for $v$.




\section{Acyclic colourings such that each colour class induces an acyclic graph}
\noi First we  show that any graph from ${\cal S}_5$ has an acyclic $\cd_1^{(5)}$-colouring. We start with the following auxiliary lemma.

\begin{lem}\label{lem_ac5}
Let $G\in\cs_5$ and $c$ be a partial acyclic $\cd_1^{(5)}$-colouring of $G$. Assume  $v$ is an uncoloured vertex. If $n_v\leq 4$, then there exists a colour for $v$ that allows us
to extend $c$.
\end{lem}
\begin{proof}
Let $c$ be a partial acyclic $\cd_1^{(5)}$-colouring of $G$ and $v$ be an uncoloured vertex with $n_v\leq 4$. If $n_v\leq 1$, then clearly we can colour $v$. Hence, we assume $n_v\geq 2$. 
 We show that we can colour $v$. The vertex $v$ cannot be coloured with a particular colour $i$ ($i=1,\ldots,5$) only if  there is an $(i,j)$-dangerous cycle for $v$, where $j\in\{1,\ldots,5\}, j\neq i$  or if there is an $i$-mono-dangerous cycle for $v$.  It is easy to observe the following:

\begin{pro}\label{obs_rnb}
If $c$ is any partial acyclic $\cd_1^{(5)}$-colouring of $G\in\cs_5$  and  $u\in V(G)$ is rainbow, then we can colour or recolour $u$ with any of $5$ colours. 
\end{pro}
Thus, we need to consider four cases. 

\noindent{\bf\itshape{Case 1}} Assume that exactly two of the coloured neighbours of  $v$ have the same colour, say $x,y \in N(v)$ and $c(x)=c(y)=1$, the others (if exist) have distinct colours or are uncoloured.   Observe that  for each colour $i\in\{2,\ldots,5\}$ there must be an  $(i,1)$-dangerous cycle and a 1-mono-dangerous cycle for $v$, passing through $x$, since otherwise we can colour $v$. It follows $x$  is adjacent to at least five distinct coloured vertices, but this is impossible, since $d(x)\leq 5$ and $x$ is also adjacent to $v$ (which is uncoloured).

\noindent{\bf\itshape{Case 2}} Assume that exactly three of the coloured neighbours of $v$ have the same colour, say $x,y, z\in N(v)$ and $c(x)=c(y)=c(z)=1$.  If one of $x,y,z$ has four coloured neighbours and is rainbow, then   Proposition \ref{obs_rnb} yields  we can recolour this vertex with a colour $i\not\in C_v$ and obtain Case 1. It follows  $p_u\leq 3$, for $u\in\{x,y,z\}$.  Moreover, we cannot colour $v$ only if for each colour $i\in\{2,\ldots,5\}$ there is an $(i,1)$-dangerous cycle and a 1-mono-dangerous cycle for $v$. Hence each colour $i\in\{1,\ldots,5\}$ belongs to at least two different multisets among $C_x, C_y, C_z$. Thus, $p_x+p_y+p_z\geq 10$, what is impossible.

\noindent{\bf\itshape{Case 3}} If there are two pairs (say, $x,y$ and $z,w$) of neighbours of $v$ with the same colour, w.l.o.g.  $c(x)=c(y)=1$ and $c(z)=c(w)=2$, then, similarly as above, we may assume none of $x,y,z,w$ has neighbours with four different colours. Thus, $p_u\leq 3$, for $u\in\{x,y,z,w\}$. We may assume that for each $i\in\{3,4,5\}$ there is an $(i,1)$-dangerous or an $(i,2)$-dangerous cycle for $v$  and that there is also a 1-mono-dangerous and a 2-mono-dangerous cycle for $v$. Hence each colour $i\in\{1,\ldots,5\}$ belongs to at least two different multisets among $C_x, C_y$ or to at least two different multisets among $C_z, C_w$. Thus, at least one of the following occurs: $p_x=p_y=3$ or $p_z=p_w=3$. W.l.o.g., we may assume the former holds. Notice, $c(x)=c(y)$. We focus on $x$. Assume $x_1,x_2,x_3\in N(x)\setminus \{v\}$ are all in some dangerous (or mono-dangerous) cycles for $v$ and $c(x_1)=c_1,c(x_2)=c_2, c(x_3)=c_3$, $c_a\neq c_b$ for $a\neq b$.
If $n_x=3$, then we  recolour $x$ with 2 and obtain Case 2. Otherwise, we may assume that there is a neighbour $x_4$ of $x$ such that $c(x_4)=c_3$. Observe that  we cannot recolour $x$ with a colour $i\in\{2,\ldots,5\}$ only if for each $i$ there is an alternating $(c_3,i)$-path from $x_3$ to $x_4$. Hence $2,\ldots,5\in C_{x_3}$, but this is impossible, since $C_{x_3}$ already contains colour 1 twice (because $x_3$ is in the dangerous cycle for $v$).

\noindent{\bf\itshape{Case 4}}
Assume that $n_v=4$ and all coloured neighbours of $v$ have the same colour, say $x, y,$ $z, w\in N(v)$ and $c(x)=c(y)=c(z)=c(w)=1$. As above, we may assume none of $x,y,z,w$  has neighbours with four different colours. Moreover, for each $i\in\{2,\ldots,5\}$ there must be an $(i,1)$-dangerous cycle and also a 1-mono-dangerous cycle for $v$, since otherwise we can colour $v$. Hence each $i\in\{1,\ldots, 5\}$ belongs to at least two different multisets among $C_x, C_y, C_z, C_w$. Thus,  there are at least two vertices among $x,y,z,w$, say $x, y$, such that $p_x=p_y=3$. We proceed similarly, as in Case 3.
\end{proof}

\noi To prove the next theorem we adapt the method presented in \cite{ho11}. In particular, we use a notion of good spanning trees.
Let $G$ be a $d$-regular connected graph. A \textit{good spanning tree} of $G$ is defined as its spanning tree   that  contains a vertex, called {\it root}, adjacent to $d-1$ leaves. 
 
\begin{thm}[\cite{ho11}]\label{thm_tree}
Every regular connected graph admits a good spanning tree.
\end{thm}

\begin{thm}\label{acyclic_5}
Every graph $G\in\cs_5$ has an acyclic $\cd_1^{(5)}$-colouring.
\end{thm}
\begin{proof}
Let $G\in\cs_5$.  Clearly, we can assume $G$ is connected. If $\delta(G)\leq 4$, then obviously $G$ is 4-degenerate. Thus,  Lemma \ref{lem_ac5} yields there is an acyclic $\cd_1^{(5)}$-colouring of $G$.

Now we may assume that $G$ is 5-regular. Let $T$ be a good spanning tree of $G$, rooted at $v_n$, where $n$ is the order of $G$ (the existence of such a tree follows from Theorem \ref{thm_tree}).
Let $N(v_n)=\{v_1,v_2,v_3,v_4,u\}$, where $v_1,v_2,v_3$ and $v_4$ are leaves in $T$. We order the vertices of $G$, from $v_5$ to $v_n$, according to the post order walk of $T$. 
We construct an acyclic $\cd_1^{(5)}$-colouring as follows. First, we colour the vertices $v_1,\ldots, v_4$ with four different colours. Then we  successively colour 
vertices $v_5$ up to $v_{n-1}=u$, using Lemma \ref{lem_ac5}, but we will never recolour the vertices $v_1,\ldots, v_4$. Now let us check that this is possible. Assume that we are going to colour $v_i$, where $i\in\{5,\ldots,n-1\}$. If $v_i$ is rainbow or has at most one coloured neighbour, then clearly we can colour it. Otherwise, one of the Cases 1--4, considered in Lemma \ref{lem_ac5}, occurs. In Case 1 there is no recolouring.  In Cases 2--4  it may happen, that we need to recolour a  rainbow vertex, say $w$, from the neighbourhood of $v_i$. However, we do it only if $w$ has four coloured neighbours. We claim that $w\not\in\{v_1,\ldots,v_4\}$. Indeed, if $v_j$ is adjacent to $v_i$, for some $j\in\{1,\ldots,4\}$, then $v_j$ has at most three coloured neighbours ($d(v_j)=5$ and $v_j$ has at least two uncoloured neighbours, namely $v_i$ and $v_n$).  In Cases 3 and 4 it is  also possible that we  recolour another neighbour of $v_i$, say $x$. But in this case there is always another neighbour of $v_i$, say $y$, with the same colour as $x$ such that we can recolour $y$ instead of $x$.  Hence if $x\in \{v_1,\ldots,v_4\}$, then we recolour $y$. Clearly, if $x$ is one of $v_1,\ldots,v_4$, then $y$ is not,  because $x$ and $y$ have the same colour.

Now let $c$ be the obtained partial colouring of $G$,  with $v_n$ being the only one  uncoloured vertex. If $v_n$ is rainbow, then Proposition \ref{obs_rnb}
yields we can colour $v_n$.
Otherwise, assume that there are two neighbours of $v_n$ with the same colour, say $c(v_1)=c(u)=1$,  all  other neighbours have distinct colours.  We cannot colour $v_n$ only if for each colour $\alpha\in\{2,\ldots,5\}$ there is an $(\alpha,1)$-dangerous cycle and there is a 1-mono-dangerous cycle for $v_n$. Each such cycle passes through both $v_1$ and $u$. Hence $p_{v_1}\geq 5$ and $p_{u}\geq 5$, but this is impossible, because $d(v_1)=d(u)=5$.
\end{proof}
 
\section{Acyclic colourings such that each colour class induces an acyclic graph with bounded degree}
  \noi Since $K_6$ needs 5 colours in any acyclic colouring in which each colour class induces an acyclic graph, we cannot reduce the number of colours in Theorem \ref{acyclic_5}. Nevertheless, we can improve Theorem \ref{acyclic_5} in the following way.

\begin{thm}\label{thm_s5}
Every graph $G\in\cs_5$ has an acyclic $({\cal S}_4\cap {\cal D}_1)^{(5)}$-colouring.
\end{thm}

\begin{proof}
Let $G\in {\cal S}_5$.  Theorem \ref{acyclic_5} implies that there is an acyclic $\cd_1^{(5)}$-colouring of $G$. We choose such a colouring $c$ with the smallest possible number of vertices that have 5 neighbours with its colour. Let $v$ be a vertex that has 5 neighbours coloured with $c(v)$. We  show that we can recolour $v$ or a neighbour of $v$ (sometimes we must recolour some other vertices first) in such a way that the obtained colouring is an acyclic $\cd_1^{(5)}$-colouring of $G$ with smaller number of vertices having 5 neighbours coloured with its colour. It is easily  seen that it is enough to prove the theorem for 5-regular graphs. Thus, let $G$ be 5-regular.

Assume that $c(v)=1$ and $C_v=\{1,1,1,1,1\}$. Let $N(v)=\{x,y,z,w,t\}$.  A coloured vertex $u$ is called {\it $4$-saturated}, if it has exactly $4$ neighbours coloured with $c(u)$. First observe the following:


 
\begin{claim1}\label{claim1}
Suppose that there is a neighbour of $v$, say $x$, such that each vertex of $N(x)\setminus \{v\}$ is coloured with distinct colour  and $x$ has a neighbour, say $x_1$, such that $c(x_1)\neq 1$ and $x_1$ is not $4$-saturated.  Then we can recolour $x$ with $c(x_1)$.
\end{claim1}

\noi Observe that if we  cannot recolour $v$, then for any colour $i\in \{2,\ldots, 5\}$ there is an $(i,1)$-dangerous cycle for $v$. Thus, each colour $i\in\{2,\ldots,5\}$ must be contained in at least two different multisets among $C_x,C_y,C_z,C_w,C_t$. It follows that there is a neighbour of $v$, say $x$, which belongs to at least two  dangerous cycles for $v$. We will focus on this vertex and consider all possible assignments of colours to its neighbours.   Let $N(x)=\{v,x_1,x_2,x_3,x_4\}$. We may assume w.l.o.g. that $c(x_1)=2, c(x_2)=3$ and both $x_1$ and $x_2$ belong to  some dangerous cycles for $v$. Thus we have the following

\begin{obs}\label{obs1}
In both $C_{x_1}$ and $C_{x_2}$ colour $1$ occurs at least twice; neither $x_1$ nor $x_2$ is $4$-saturated.
\end{obs}

\begin{claim1}\label{claim2}
Suppose that   in $C(N(x)\setminus \{v\})$ exactly one colour  occurs twice. Then we can recolour $x$.
\end{claim1}

\begin{proof}
Let $c(x_3)=a, c(x_4)=b$. Since exactly one colour in $C(N(x)\setminus \{v\})$ occurs twice, there are (up to symmetry) three cases to consider.

\ms
 
\noindent{\bf\textit{Case 1}} Assume that $a=b=1$.

\noindent If we cannot recolour $x$ with any colour $i$, for $i\in\{2,3,4,5\}$, then for each $i\in\{2,3,4,5\}$ there is an $(i,1)$-dangerous cycle for $x$, passing through both $x_3$ and $x_4$. Thus, $C_{x_3}=\{1,\ldots,5\}$. Further, none of the neighbours of $x_3$ is 4-saturated. Hence, we can recolour $x_3$ with 4. Thus, by Claim \ref{claim1} we can recolour $x$.

\noindent{\bf\textit{Case 2}} Suppose that $a\in\{1,4,5\}$ and $b\in\{2,3\}$.

\noindent For convenience, we may assume that $b=2$ and $a\neq 5$. If we cannot recolour $x$ neither with 3 nor with 5, then there is an $(i,2)$-dangerous cycle for $x$, passing through both $x_1$ and $x_4$, for each  $i\in\{3,5\}$. Thus, $x_4$ is not 4-saturated. Next, if we cannot recolour $x$ with 2, then there must be a 2-mono-dangerous cycle for $x$, passing through both $x_1$ and $x_4$. It follows that $C_{x_1}=\{1,1,2,3,5\}$. We may recolour $x_1$ with 5 (since the neighbour of $x_1$ coloured with 5 is not 4-saturated). Claim \ref{claim1} implies we can recolour $x$.  

\noindent{\bf\textit{Case 3}} Let $b=a$ and $a\in\{4,5\}$.

\noindent Assume w.l.o.g. that $a=4$. If we cannot recolour $x$ with any colour $i$, for $i\in\{2,3,5\}$, then for each $i\in\{2,3,5\}$ there is an $(i,4)$-dangerous cycle for $x$, passing through both $x_3$ and $x_4$. Thus, neither $x_3$ nor $x_4$ is 4-saturated. It follows that if we cannot recolour $x$ with 4, then there must be a 4-mono-dangerous cycle for $x$, passing through both $x_3$ and $x_4$. Hence $C_{x_3}=\{1,\ldots,5\}$. We may recolour $x_3$ with 5 (since the neighbour of $x_3$ coloured with 5 is not 4-saturated). From Claim \ref{claim1} it follows that we can recolour $x$.
\end{proof}

\noindent Claims \ref{claim1} and \ref{claim2} imply  that we have to consider only two cases:  both colours $2,3$ occur twice in $C_x$ and one of colours from $\{2,3\}$ occurs three times.

\ms
 
\noindent{\bf\textit{Case 1}} Assume  that $c(x_1)=2, c(x_2)=3,c(x_3)=2, c(x_4)=3$. 

\noindent Observe that:
\begin{claim1}\label{claim3}
If there is a neighbour $x_i$ of $x$  such that in $C(N(x_i)\setminus \{x\})$  each colour occurs exactly once and at most one vertex of $N(x_i)\setminus \{x\}$ is both $4$-saturated and coloured with either $4$ or $5$, then we can recolour $x$.
\end{claim1}

\begin{proof}
First, we recolor $x_i$ either with 4 or 5. Since at most one vertex of $N(x_i)\setminus \{x\}$ coloured with 4 or 5 is 4-saturated and $4,5\notin C_x$, such a recolouring is possible. Thus, we obtain the colouring that satisfies the assertions of Claim \ref{claim2} and hence we can recolour $x$.
\end{proof}

\begin{claim1}\label{claim4}
If $x$ has a neighbour $x_i$ such that  $x_i$ does not have a $4$-saturated neighbour coloured with $1$ and  there is neither a $1$-mono-dangerous cycle  nor a $(1,j)$-dangerous cycle for $x_i$ $($for each $j\in\{2,3,4,5\})$ passing through two vertices of $N(x_i)\setminus \{x\}$, then we can recolour $x$.
\end{claim1}

\begin{proof}
Suppose to the contrary that there is a neighbour $x_i$ of $x$ such that none of the coloured 1 neighbours of $x_i$ is 4-saturated and there is neither a 1-mono-dangerous cycle  nor a $(1,j)$-dangerous cycle for $x_i$ passing through two vertices of $N(x_i)\setminus \{x\}$ and we cannot recolour $x$. If we can recolour $x_i$ with 1, then the colouring satisfies the assertions of Claim \ref{claim2} and hence we can recolour $x$, a contradiction. Otherwise, there is a 1-mono-dangerous cycle for  $x_i$ passing through $x$. For convenience and w.l.o.g. we may assume  that $c(x_i)=2$. (Thus, $i\in\{1,3\}$.) If there is no $(a,3)$-dangerous cycle for $x$ ($a\in\{4,5\}$), then we may recolour $x_i$ with 1 and $x$ with $a$. Thus, we  assume that both such dangerous cycles are present, and are passing through $x_2$ and $x_4$. Hence, $x_4$ is not 4-saturated. (Recall that Observation \ref{obs1} yields that also $x_2$ is not 4-saturated.) It follows that we cannot recolour $x$ with 3 only if there is a 3-mono-dangerous cycle for $x$, passing through $x_2$ and $x_4$. From the above and Observation \ref{obs1} we clearly obtain  $C_{x_2}=\{1,1,3,4,5\}$. Therefore $x_2$ satisfies the assertions of Claim \ref{claim3}. This implies that we can recolour $x$.
\end{proof}

\noindent Observe that if we can recolour $x$ with 4 or 5, then we are done. This is impossible only if there are at least two dangerous cycles for $x$, namely a $(4,i)$-dangerous cycle and a $(5,j)$-dangerous cycle, where $i,j \in\{2,3\}$. Up to symmetry, there are two possibilities.

\sm
\begin{figure}[htbp]
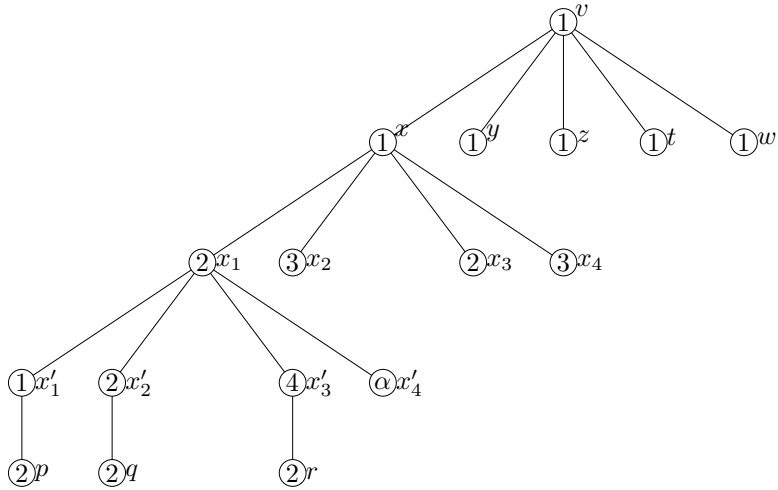

\centering

\grafGR

	\caption{Subcase 1.1. Since there is a $(2,1)$-dangerous cycle for $v$, $c(p)=2$. Since there is a $2$-mono-dangerous cycle for $x$, $c(q)=2$. Since there is a $(4,2)$-dangerous cycle for $x$, $c(r)=2$.  }

\end{figure}

\noindent{\bf\textit{Subcase 1.1}} Assume that both a $(4,2)$-dan\-ger\-ous cycle and a $(5,3)$-dangerous cycle for $x$ are present. It follows $4\in C_{x_1}$, $4\in C_{x_3}$ and $5\in C_{x_2}$, $5\in C_{x_4}$. Next, we cannot recolour $x$ with 2  only if there is a $(2,3)$-dangerous cycle for $x$, passing through $x_2$ and $x_4$, or there is a 2-mono-dangerous cycle for $x$, passing through $x_1$ and $x_3$. The same argument can be applied if we consider recolouring $x$ with 3. 
Assume first that  we have for $x$ both a 2-mono-dangerous cycle and a $(3,2)$-dangerous cycle, both passing through $x_1$. Recall that Observation \ref{obs1} yields that colour 1 occurs twice in  $C_{x_1}$. Thus $C_{x_1}=\{1, 1,2, 3,4\}$ and  by Claim \ref{claim3} we can recolour $x$, so we are done. We can proceed in this way also in the case when we have for $x$ both a 3-mono-dangerous cycle and a $(2,3)$-dangerous cycle (both passing through $x_2$). 
Hence we have two possibilities left. There are either both a 2-mono and a 3-mono-dangerous  cycle for $x$, or both a $(2,3)$-dangerous  and a $(3,2)$-dangerous cycle for $x$. We consider only the first situation, since the second one can be solved analogously. 
It follows $C_{x_1}=\{1,1,2,4,\alpha\}$ and $C_{x_2}=\{1,1,3,5,\beta\}$.  We focus on the vertex $x_1$ (see Figure 1). Let $N(x_1)\setminus \{x\}=\{x'_1,x'_2,x'_3,x'_4\}$ and $c(x'_1)=1, c(x'_2)=2,c(x'_3)=4, c(x'_4)=\alpha $. By Claim \ref{claim3} we may assume that $\alpha \in\{1,2,4\}$.  Claim \ref{claim4} implies that there is a $(1,\alpha)$-dangerous or 1-mono-dangerous cycle for $x_1$, passing through two vertices of $\{x'_1,x'_2,x'_3,x'_4\}$ (these two vertices are coloured with $\alpha $) or $\alpha=1$ and $x'_4$ is 4-saturated. In the former case, if recolouring $x_1$ with $4,5$ is impossible, then there are two dangerous cycles for $x_1$ passing through two vertices of $\{x'_1,x'_2,x'_3,x'_4\}$ coloured with $\alpha$. Let $x'\in \{x'_1,x'_2,x'_3\}$ be the vertex with colour $\alpha$. Thus, $C_{x'}=\{2,1,2,4,5\}$. We can recolour $x'$ with 5 (since the neighbour of $x'$ coloured with 5 is not 4-saturated) and hence by Claim \ref{claim3} we can recolour $x$. In the latter case, when $\alpha=1$ and $x'_4$ is 4-saturated, it is easy to observe that we can recolour $x_1$ with 4 and then  Claim \ref{claim2} implies we can recolour $x$.

\begin{figure}[htbp]
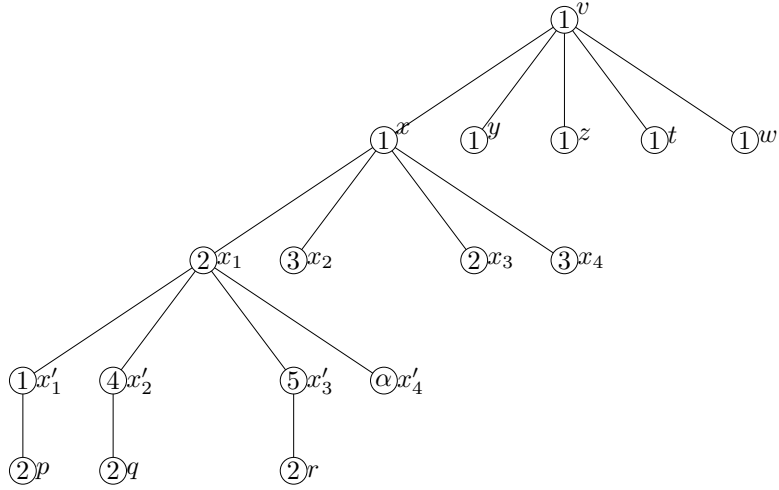

\centering

\grafGS

	\caption{Subcase 1.2. Since there is a $(2,1)$-dangerous cycle for $v$, $c(p)=2$. Since there is a a $(4,2)$-dangerous cycle for $x$, $c(q)=2$. Since there is a $(5,2)$-dangerous cycle for $x$, $c(r)=2$. }

\end{figure}

\noindent{\bf\textit{Subcase 1.2}} Suppose that there are both a  $(4,2)$-dangerous cycle and a $(5,2)$-dangerous cycle for $x$. It follows $C_{x_1}=\{1,1,4,5,\alpha\}$ (see Figure 2). Let $N(x_1)\setminus \{x\}=\{x'_1,x'_2,x'_3,x'_4\}$ and $c(x'_1)=1, c(x'_2)=2,c(x'_3)=4, c(x'_4)=\alpha $. By Claim \ref{claim3} we may assume that $\alpha \in\{1,2,4\}$. If recolouring $x_1$ with $1,4,5$ is impossible, then there are three dangerous cycles for $x_1$ passing through two vertices of $\{x'_1,x'_2,x'_3,x'_4\}$  coloured with $\alpha $ (Claim \ref{claim4} implies that the $(1,\alpha)$-dangerous or 1-mono-dangerous cycle exists) or $\alpha=1$ and $x'_4$ is 4-saturated. We start with considering the first situation.  Let $x'\in \{x'_1,x'_2,x'_3\}$ be the vertex with colour $\alpha$. Thus, $C_{x'}=\{2,1,2,4,5\}$. We can recolour $x'$ with 3 and hence by Claim \ref{claim3} we can recolour $x$. In the second case, the fact that $x'_4$ is 4-saturated implies we can recolour $x_1$ with 4 and, by Claim \ref{claim2}, we can recolour $x$.

\medskip
\noindent{\bf\textit{Case 2}} Assume  that $c(x_1)=2, c(x_2)=3,c(x_3)=3, c(x_4)=3$. 

\noindent Similarly as Claims 3 and 4 we can prove the following:

\begin{claim1}\label{claim5}
If there is a neighbour $x_i\;(i\in\{2,3 ,4\})$ of $x$  such that in $C(N(x_i)\setminus \{x\})$  each colour occurs exactly once and at most one vertex of $N(x_i)\setminus \{x\}$ is both $4$-saturated and coloured with either $4$ or $5$, then we can recolour $x$.
\end{claim1}

\begin{claim1}\label{claim6}
If $x$ has a neighbour $x_i\;(i\in\{2,3 ,4\})$ such that  $x_i$ does not have a $4$-saturated neighbour coloured with $1$ and  there is neither a $1$-mono-dangerous cycle  nor a $(1,j)$-dangerous cycle for $x_i$ passing through two vertices of $N(x_i)\setminus \{x\}$, then we can recolour $x$.
\end{claim1}

\begin{figure}[htbp]
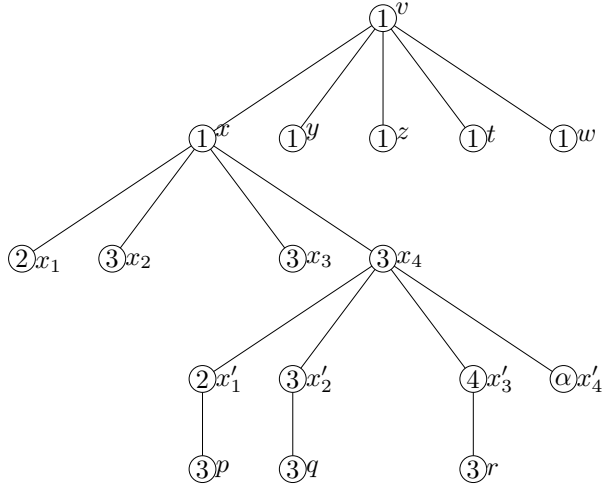

\centering

\grafGT

	\caption{Case 2.  Since there is a $(2,3)$-dangerous cycle for $x$, $c(p)=3$. Since there is  a $3$-mono-dangerous cycle for $x$, $c(q)=3$. Since there is a $(4,3)$-dangerous cycle for $x$, $c(r)=3$. }

\end{figure}

\noindent If we can recolour $x$ with 2, 4 or 5, then we are done. This is impossible only if there is an $(i,3)$-dangerous cycle for $x$, for each $i\in\{2,4,5\}$. Thus, recolouring $x$ with 3 is impossible if we have a 3-mono-dangerous cycle for $x$ or one of $x_2,x_3,x_4$ is 4-saturated. Assume at the beginning that we have a 4-saturated neighbour of $x$, say $x_4$. Hence $C_{x_2}=\{1,1,2,4,5\}$ and we can recolour $x_2$ with 4, so by Claim \ref{claim2} we can  recolour $x$. Therefore, we may assume there is a 3-mono-dangerous cycle for $x$ and none of $x_2,x_3,x_4$ is 4-saturated. Thus, each colour $i\in\{2,3,4,5\}$ occurs twice among $C_{x_2}, C_{x_3}, C_{x_4}$ and hence one multiset contains at least three of these colours.  W.l.o.g. we may assume that $C_{x_4}=\{1,2,3,4,\alpha\}$ (see Figure 3). 
By Claim \ref{claim5} we may assume that $\alpha =2,3$ or 4. If recolouring $x_4$ with $1,4,5$ is impossible, then there are three dangerous cycles for $x_4$ passing through two vertices of $\{x'_1,x'_2,x'_3,x'_4\}$  coloured with $\alpha $ (Claim \ref{claim6} implies that the $(1,\alpha)$-dangerous cycle exists). Let $x'\in \{x'_1,x'_2,x'_3\}$ be the vertex with colour $\alpha$. Thus, $C_{x'}=\{3,1,3,4,5\}$. We can recolour $x'$ with 5 and hence by Claim \ref{claim5} we can recolour $x$.
\end{proof}

\section{Complexity result}

\noindent Now we show that the problem of deciding whether a graph $G\in \cs_5$ has an acyclic  $(\cs_3,\cs_3)$-colouring
is NP-complete.  First we present some special graphs and their properties. Let $G(C_j)$ and $F$ be the graphs  depicted in Figure \ref{fig_gcj}. It is easy to observe that both 
 $G(C_j)$ and $F$ belong to $\cs_5$. Their  acyclic $(\cs_3,\cs_3)$-colourings are presented in Figure \ref{fig_gcj}. The following two observations concerning graph $G(C_j)$ are straightforward.

\newcommand{\grafGB}{
  \psset{unit=1.0cm}  
  \begin{pspicture}(-0.2,-0.2)(3.55,3.65)
    \pnode(1.75,2.80){v1}
    \pnode(0.10,0.70){v2}
    \pnode(1.30,1.60){v3}
    \pnode(1.75,1.00){v4}
    \pnode(3.25,0.70){v5}
    \pnode(2.20,1.60){v6}
    \pnode(2.20,0.10){v7}
    \pnode(1.15,2.95){v8}
    \pnode(0.55,2.65){v9}
    \pnode(0.55,1.90){v10}
    \psline[linecolor=black,linewidth=0.3pt](v2)(v1)
    \psline[linecolor=black,linewidth=0.3pt](v3)(v1)
    \psline[linecolor=black,linewidth=0.3pt](v3)(v2)
    \psline[linecolor=black,linewidth=0.3pt](v4)(v1)
    \psline[linecolor=black,linewidth=0.3pt](v4)(v2)
    \psline[linecolor=black,linewidth=0.3pt](v4)(v3)
 
    \psline[linecolor=black,linewidth=0.3pt](v5)(v2)
    \psline[linecolor=black,linewidth=0.3pt](v5)(v4)
    \psline[linecolor=black,linewidth=0.3pt](v6)(v1)
    \psline[linecolor=black,linewidth=0.3pt](v6)(v3)
    \psline[linecolor=black,linewidth=0.3pt](v6)(v4)
    \psline[linecolor=black,linewidth=0.3pt](v6)(v5)
    \psline[linecolor=black,linewidth=0.3pt](v7)(v2)
    \psline[linecolor=black,linewidth=0.3pt](v7)(v6)
    \psline[linecolor=black,linewidth=0.3pt](v8)(v1)
    \psline[linecolor=black,linewidth=0.3pt](v9)(v8)
    \psline[linecolor=black,linewidth=0.3pt](v10)(v3)
    \psline[linecolor=black,linewidth=0.3pt](v10)(v9)
    \pscircle[fillcolor=black, fillstyle=solid,linewidth=0.2pt](v1){0.05}
      \uput[79.99](v1){$_{d}$}
    \pscircle[fillcolor=white, fillstyle=solid,linewidth=0.2pt](v2){0.05}
      \uput[-149.93](v2){$_{e}$}
    \pscircle[fillcolor=black, fillstyle=solid,linewidth=0.2pt](v3){0.05}
      \uput[161.57](v3){$_{a}$}
    \pscircle[fillcolor=black, fillstyle=solid,linewidth=0.2pt](v4){0.05}
      \uput[-66.80](v4){$_{c}$}
    \pscircle[fillcolor=white, fillstyle=solid,linewidth=0.2pt](v5){0.05}
      \uput[-25.56](v5){$_{g}$}
    \pscircle[fillcolor=black, fillstyle=solid,linewidth=0.2pt](v6){0.05}
      \uput[6.34](v6){$_{b}$}
    \pscircle[fillcolor=white, fillstyle=solid,linewidth=0.2pt](v7){0.05}
      \uput[-64.65](v7){$_{x'}$}
    \pscircle[fillcolor=white, fillstyle=solid,linewidth=0.2pt](v8){0.05}
      \uput[104.74](v8){$_{i}$}
    \pscircle[fillcolor=white, fillstyle=solid,linewidth=0.2pt](v9){0.05}
      \uput[130.91](v9){$_{x}$}
    \pscircle[fillcolor=white, fillstyle=solid,linewidth=0.2pt](v10){0.05}
      \uput[158.96](v10){$_{h}$}
  \end{pspicture}
}

\newcommand{\grafGC}{
  \psset{unit=1.0cm}  
  \begin{pspicture}(-0.2,-0.2)(2.80,4.55)
    \pnode(1.90,3.85){v1}
    \pnode(1.60,2.80){v2}
    \pnode(0.10,2.05){v3}
    \pnode(0.85,2.05){v4}
    \pnode(1.60,1.30){v5}
    \pnode(1.90,0.10){v6}
    \pnode(2.50,2.05){v7}
    \psline[linecolor=black,linewidth=0.3pt](v3)(v1)
    \psline[linecolor=black,linewidth=0.3pt](v3)(v2)
    \psline[linecolor=black,linewidth=0.3pt](v4)(v1)
    \psline[linecolor=black,linewidth=0.3pt](v4)(v2)
    \psline[linecolor=black,linewidth=0.3pt](v5)(v2)
    \psline[linecolor=black,linewidth=0.3pt](v5)(v3)
    \psline[linecolor=black,linewidth=0.3pt](v5)(v4)
    \psline[linecolor=black,linewidth=0.3pt](v6)(v1)
    \psline[linecolor=black,linewidth=0.3pt](v6)(v3)
    \psline[linecolor=black,linewidth=0.3pt](v6)(v4)
    \psline[linecolor=black,linewidth=0.3pt](v7)(v1)
    \psline[linecolor=black,linewidth=0.3pt](v7)(v2)
    \psline[linecolor=black,linewidth=0.3pt](v7)(v5)
    \psline[linecolor=black,linewidth=0.3pt](v7)(v6)
    \pscircle[fillcolor=black, fillstyle=solid,linewidth=0.1pt](v1){0.05}
      \uput[64.36](v1){$$}
    \pscircle[fillcolor=black, fillstyle=solid,linewidth=0.1pt](v2){0.05}
      \uput[53.97](v2){$$}
    \pscircle[fillcolor=black, fillstyle=solid,linewidth=0.1pt](v3){0.05}
      \uput[175.24](v3){$_{u_{j,1}}$}
    \pscircle[fillcolor=black, fillstyle=solid,linewidth=0.1pt](v4){0.05}
      \uput[153.43](v4){$_{u_{j,2}}$}
    \pscircle[fillcolor=black, fillstyle=solid,linewidth=0.1pt](v5){0.05}
      \uput[-48.37](v5){$$}
    \pscircle[fillcolor=white, fillstyle=solid,linewidth=0.1pt](v6){0.05}
      \uput[-64.36](v6){$$}
    \pscircle[fillcolor=white, fillstyle=solid,linewidth=0.1pt](v7){0.05}
      \uput[2.86](v7){$_{u_{j,3}}$}
  \end{pspicture}
}

\begin{figure}[htbp]
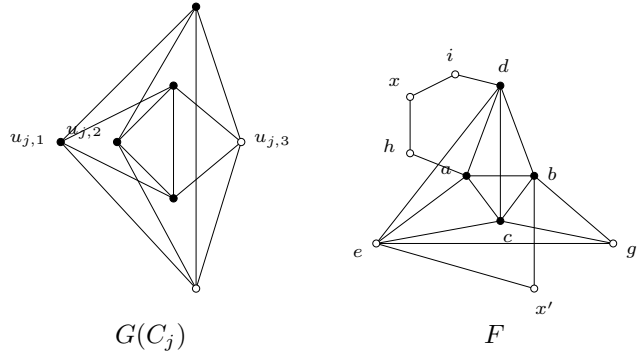

\centering
\begin{tabular}{ccc}
\grafGC &\hspace{10pt} & \grafGB\\
$G(C_j)$ &\hspace{10pt} & $F$\\
\end{tabular}

	\caption{Graphs $G(C_j)$ and $F$.}
	\label{fig_gcj}

\end{figure}

 \begin{obs} \label{c_jkolor}
In any acyclic  $(\cs_3,\cs_3)$-colouring of   $G(C_j)$  the vertices 
 $u_{j,1},$ $u_{j,2},$ $u_{j,3}$ are not all coloured with the same colour.
\end{obs}

\begin{obs}\label{c_j}
Any partition of the set $\{u_{j,1},$ $u_{j,2},$ $u_{j,3}\}$  into two nonempty parts can be extended to an acyclic $(\cs_3,\cs_3)$-colouring of  $G(C_j)$.
\end{obs}
 
\noindent A coloured vertex $v$ is called {\it $3$-saturated}, if it has exactly $3$ neighbours coloured with $c(v)$. Considering the graph $F$, its important properties are the following.
 
\begin{obs}\label{kolF}
 In any acyclic $(\cs_3,\cs_3)$-colour\-ing $f$ of  $F$, both $x$ and $x'$ have the same colour, $x$ has exactly two neighbours coloured with $f(x)$, $x'$ has exactly one neighbour coloured with $f(x')$.
\end{obs}
\begin{proof}
Let $f$ be an acyclic $(\cs_3,\cs_3)$-colouring of $F$. We prove that $f$ has the desired properties, by considering all possible colourings of the vertices $a,b,c$.  

\noindent{\bf\textit{Case} 1} Assume that $a,b,c$ all have the same colour, say 1.

If $d$ also has colour 1, then $i,h,e,g$ and $x'$ must be coloured with 2, because each of $a,b,c,d$ is 3-saturated. 
 There is an alternating path between  $h$ and $i$, hence  $x$ must be coloured with 2.
Clearly, such a colouring has the desired properties.

Suppose now that $d$ has colour 2. It follows   $e$ and $g$ must be coloured with 1, since otherwise we have an alternating cycle. Thus, $c$ has four neighbour in its own colour, a contradiction.

\noindent{\bf\textit{Case} 2} Suppose two vertices among  $a,b,c$ are coloured with 1, the remaining one is coloured with 2. There are three possibilities. 

Assume $a,b$ have colour 1, $c$ has colour 2. Observe that $d,e,g$ must have colour 1, since otherwise an alternating cycle occurs. Hence $b$ is 3-saturated, thus 
$x'$ must be coloured with 2. It follows that there is an alternating cycle induced by $x',b,c,e$, a contradiction.

Suppose now that $a,c$ have colour 1, $b$ has colour 2. Clearly,  $d,e$ must have colour 1, otherwise an alternating cycle occurs. Thus, $d$ is 3-saturated. Hence 
 $i$ must be coloured with 2. Further, each of $c,e,a$ is 3-saturated, thus 
 $g, x',h$ must have colour 2.
There is an alternating path between  $h$ and $i$, hence $x$ must have colour 2. Again, such a colouring has the desired properties.

Finally, assume  $b,c$ have colour 1, $a$ has colour 2. It follows $d,g$ must be coloured with 1, otherwise we would have an alternating cycle. Thus, $c$ is 3-saturated, therefore $e$ must have colour 2, but then vertices $d,a,c,e$ induce an alternating cycle, a contradiction follows.
\end{proof}

\begin{obs}\label{sciezkaF}
In any acyclic  $(\cs_3,\cs_3)$-colour\-ing of $F$
there is no alternating path between the vertices $x$ and $x'$. 
\end{obs}

\noindent 
Let $m$ be a positive integer. Based on the graph $F$, we construct a graph  $G^m(v_i)$.  We take
 $m$ copies of $F$, say $F_1,F_2,\ldots,F_m$. In each copy $F_s$,  vertices 
$x$ and $x'$ are denoted by  $x_{i,s}$ and  $x'_{i,s}$, resp. Next, for $s\in\{1,\ldots,m-1\}$, we identify  
 $x'_{i,s}$ with $x_{i,s+1}$,  and denote the obtained vertex  by   $v_{i,s+1}$.
Finally, we identify $x_{i,1}$ with $x'_{i,m}$ and denote the new vertex by $v_{i,1}$. 
Observe that  $G^m(v_i)\in\cs_5$ and  the vertices $v_{i,s}$ are all of degree 4. Such a graph for $m=3$ is presented in Figure \ref{fig_gvi}.

\newcommand{\grafGe}{
  \psset{unit=0.5cm}  
  \begin{pspicture}(-0.2,-3.2)(11.60,5.00)
    \pnode(1.70,2.10){v1}
    \pnode(0.10,0.10){v2}
    \pnode(1.30,0.90){v3}
    \pnode(1.70,0.50){v4}
    \pnode(3.10,0.10){v5}
    \pnode(2.10,0.90){v6}
    \pnode(1.10,2.30){v7}
    \pnode(0.30,2.10){v8}
    \pnode(0.50,1.30){v9}
    
    \pnode(3.70,-1.00){v10}
    \pnode(4.50,-0.70){v11}
    \pnode(5.10,-0.90){v12}
    \pnode(3.90,-1.70){v13}
    \pnode(4.70,-2.10){v14}
    \pnode(5.50,-2.10){v15}
    \pnode(5.10,-2.50){v16}
    \pnode(3.30,-2.90){v17}
    \pnode(6.70,-2.90){v18}
    
    \pnode(8.70,2.10){v19}
    \pnode(9.50,2.30){v20}
    \pnode(10.10,2.10){v21}
    \pnode(8.90,1.30){v22}
    \pnode(9.70,0.90){v23}
    \pnode(10.50,0.90){v24}
    \pnode(10.10,0.50){v25}
    \pnode(8.30,0.10){v26}
    \pnode(11.70,0.10){v27}
    
    \psline[linecolor=black,linewidth=0.2pt](v2)(v1)
    \psline[linecolor=black,linewidth=0.2pt](v3)(v1)
    \psline[linecolor=black,linewidth=0.2pt](v3)(v2)
    \psline[linecolor=black,linewidth=0.2pt](v4)(v1)
    \psline[linecolor=black,linewidth=0.2pt](v4)(v2)
    \psline[linecolor=black,linewidth=0.2pt](v4)(v3)
    \psline[linecolor=black,linewidth=0.2pt](v5)(v2)
    \psline[linecolor=black,linewidth=0.2pt](v5)(v4)
    \psline[linecolor=black,linewidth=0.2pt](v6)(v1)
    \psline[linecolor=black,linewidth=0.2pt](v6)(v3)
    \psline[linecolor=black,linewidth=0.2pt](v6)(v4)
    \psline[linecolor=black,linewidth=0.2pt](v6)(v5)
    \psline[linecolor=black,linewidth=0.2pt](v7)(v1)
    \psline[linecolor=black,linewidth=0.2pt](v8)(v7)
    \psline[linecolor=black,linewidth=0.2pt](v9)(v3)
    \psline[linecolor=black,linewidth=0.2pt](v9)(v8)
    
    \psline[linecolor=red,linewidth=0.2pt](v10)(v2)
    \psline[linecolor=red,linewidth=0.2pt](v10)(v6)
    
    \psline[linecolor=black,linewidth=0.2pt](v11)(v10)
    \psline[linecolor=black,linewidth=0.2pt](v12)(v11)
    \psline[linecolor=black,linewidth=0.2pt](v13)(v10)
    \psline[linecolor=black,linewidth=0.2pt](v14)(v12)
    \psline[linecolor=black,linewidth=0.2pt](v14)(v13)
    \psline[linecolor=black,linewidth=0.2pt](v15)(v12)
    \psline[linecolor=black,linewidth=0.2pt](v15)(v14)
    \psline[linecolor=black,linewidth=0.2pt](v16)(v12)
    \psline[linecolor=black,linewidth=0.2pt](v16)(v14)
    \psline[linecolor=black,linewidth=0.2pt](v16)(v15)
    \psline[linecolor=black,linewidth=0.2pt](v17)(v12)
    \psline[linecolor=black,linewidth=0.2pt](v17)(v14)
    \psline[linecolor=black,linewidth=0.2pt](v17)(v16)
    \psline[linecolor=black,linewidth=0.2pt](v18)(v15)
    \psline[linecolor=black,linewidth=0.2pt](v18)(v16)
    \psline[linecolor=black,linewidth=0.2pt](v18)(v17)
    
    \psline[linecolor=black,linewidth=0.2pt](v20)(v19)
    \psline[linecolor=black,linewidth=0.2pt](v21)(v20)
    \psline[linecolor=black,linewidth=0.2pt](v22)(v19)
    \psline[linecolor=black,linewidth=0.2pt](v23)(v21)
    \psline[linecolor=black,linewidth=0.2pt](v23)(v22)
    \psline[linecolor=black,linewidth=0.2pt](v24)(v21)
    \psline[linecolor=black,linewidth=0.2pt](v24)(v23)
    \psline[linecolor=black,linewidth=0.2pt](v25)(v21)
    \psline[linecolor=black,linewidth=0.2pt](v25)(v23)
    \psline[linecolor=black,linewidth=0.2pt](v25)(v24)
    \psline[linecolor=black,linewidth=0.2pt](v26)(v21)
    \psline[linecolor=black,linewidth=0.2pt](v26)(v23)
    \psline[linecolor=black,linewidth=0.2pt](v26)(v25)
    \psline[linecolor=black,linewidth=0.2pt](v27)(v24)
    \psline[linecolor=black,linewidth=0.2pt](v27)(v25)
    \psline[linecolor=black,linewidth=0.2pt](v27)(v26)
    
    \psline[linecolor=red,linewidth=0.2pt](v19)(v15)
    \psline[linecolor=red,linewidth=0.2pt](v19)(v17)
    
    \psline[linecolor=red,linewidth=0.2pt](v8)(v24)
    \psline[linecolor=red,linewidth=0.2pt](v8)(v26)
    
    \pscircle[fillcolor=black, fillstyle=solid,linewidth=0.1pt](v1){0.05}
      \uput[152.10](v1){$$}
    \pscircle[fillcolor=black, fillstyle=solid,linewidth=0.1pt](v2){0.05}
      \uput[-161.57](v2){$$}
    \pscircle[fillcolor=black, fillstyle=solid,linewidth=0.1pt](v3){0.05}
      \uput[-171.87](v3){$$}
    \pscircle[fillcolor=black, fillstyle=solid,linewidth=0.1pt](v4){0.05}
      \uput[-157.62](v4){$$}
    \pscircle[fillcolor=black, fillstyle=solid,linewidth=0.1pt](v5){0.05}
      \uput[-105.26](v5){$$}
    \pscircle[fillcolor=black, fillstyle=solid,linewidth=0.1pt](v6){0.05}
      \uput[-167.01](v6){$$}
    \pscircle[fillcolor=black, fillstyle=solid,linewidth=0.1pt](v7){0.05}
      \uput[154.44](v7){$$}
    \pscircle[fillcolor=black, fillstyle=solid,linewidth=0.1pt](v8){0.05}
      \uput[160.46](v8){$_{v_{i,1}}$}
    \pscircle[fillcolor=black, fillstyle=solid,linewidth=0.1pt](v9){0.05}
      \uput[178.03](v9){$$}
    \pscircle[fillcolor=black, fillstyle=solid,linewidth=0.1pt](v10){0.05}
      \uput[65.53](v10){$_{v_{i,2}}$}
    \pscircle[fillcolor=black, fillstyle=solid,linewidth=0.1pt](v11){0.05}
      \uput[36.25](v11){$$}
    \pscircle[fillcolor=black, fillstyle=solid,linewidth=0.1pt](v12){0.05}
      \uput[23.20](v12){$$}
    \pscircle[fillcolor=black, fillstyle=solid,linewidth=0.1pt](v13){0.05}
      \uput[6.34](v13){$$}
    \pscircle[fillcolor=black, fillstyle=solid,linewidth=0.1pt](v14){0.05}
      \uput[-10.01](v14){$$}
    \pscircle[fillcolor=black, fillstyle=solid,linewidth=0.1pt](v15){0.05}
      \uput[-6.84](v15){$$}
    \pscircle[fillcolor=black, fillstyle=solid,linewidth=0.1pt](v16){0.05}
      \uput[-18.43](v16){$$}
    \pscircle[fillcolor=black, fillstyle=solid,linewidth=0.1pt](v17){0.05}
      \uput[-74.74](v17){$$}
    \pscircle[fillcolor=black, fillstyle=solid,linewidth=0.1pt](v18){0.05}
      \uput[-16.56](v18){$$}

 \pscircle[fillcolor=black, fillstyle=solid,linewidth=0.1pt](v19){0.05}
      \uput[57.53](v19){$_{v_{i,3}}$}
    \pscircle[fillcolor=black, fillstyle=solid,linewidth=0.1pt](v20){0.05}
      \uput[36.25](v20){$$}
    \pscircle[fillcolor=black, fillstyle=solid,linewidth=0.1pt](v21){0.05}
      \uput[23.20](v21){$$}
    \pscircle[fillcolor=black, fillstyle=solid,linewidth=0.1pt](v22){0.05}
      \uput[6.34](v22){$$}
    \pscircle[fillcolor=black, fillstyle=solid,linewidth=0.1pt](v23){0.05}
      \uput[-10.01](v23){$$}
    \pscircle[fillcolor=black, fillstyle=solid,linewidth=0.1pt](v24){0.05}
      \uput[-6.84](v24){$$}
    \pscircle[fillcolor=black, fillstyle=solid,linewidth=0.1pt](v25){0.05}
      \uput[-18.43](v25){$$}
    \pscircle[fillcolor=black, fillstyle=solid,linewidth=0.1pt](v26){0.05}
      \uput[-74.74](v26){$$}
    \pscircle[fillcolor=black, fillstyle=solid,linewidth=0.1pt](v27){0.05}
      \uput[-16.56](v27){$$}
  \end{pspicture}
}

\begin{figure}[htbp]
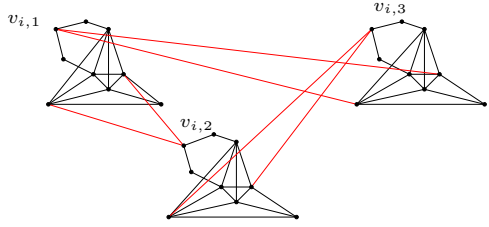

\centering
\grafGe
	\caption{Graph $G^3(v_i)$.}
	\label{fig_gvi}
	\centering
\end{figure}

\begin{obs}\label{g_vi}
For any positive integer $m$, the graph $G^m(v_i)$ has an acyclic $(\cs_3,\cs_3)$-colouring. In any such  a colouring the vertices  $v_{i,s}$ are 3-saturated and are all in the same colour, for $s\in\{1,\ldots,m\}$.
\end{obs}
\begin{proof}
Follows from the construction of  $G^m(v_i)$ and Observations \ref{kolF} and \ref{sciezkaF}. 
\end{proof}


\begin{thm}\label{thm_np}
The problem of deciding whether a  graph admits an acyclic $(\cs_3,\cs_3)$-colouring is NP-complete, even for graphs with maximum degree at most $5$.  
\end{thm}
\begin{proof}
Clearly the problem of acyclic $(\cs_3,\cs_3)$-colouring is in NP. We will show that it is NP-complete. We use the reduction from the 3NN-SAT problem, whose NP-completeness  was proved in \cite{gajo79}. Let $\cc = \{C_1,\ldots,C_m\}$ be
a set of clauses and let ${\mathcal V}=\{v_{1},\ldots,v_{n}\}$ be a set of Boolean variables. Furthermore, $C_{j}\subseteq {\mathcal V}$ and for $j\in\{1,\ldots,m\}$, $C_{j}=(c_{j1},c_{j2},c_{j3})$. In 3NN-SAT we ask if there is a truth assignment such that in each clause there exist at least one true variable and at least one false variable. Note that there are no negative variables.

 Given an instance $\cc$ of 3NN-SAT we create an instance of our problem, i.e., the graph $G$, in the following way. 
For each variable $v_i\in {\mathcal V}$ we take the graph $G^m(v_i)$, for each  clause $C_j\in \cc$
we take the graph $G(C_j)$. To obtain the graph $G$ we  connect graphs $G^m(v_i)$ and $G(C_j)$ in such a way that for $i=1,\ldots,n$, $j=1,\ldots, m$, $k=1,2,3$ we add an edge  $v_{i,j}u_{j,k}$ 
if and only if $c_{jk}$ is the variable $v_i$. Observe that for $j\in\{1,\ldots,m\}, k\in\{1,2,3\}$  each vertex $u_{j,k}$ is adjacent to exactly one vertex from graphs $G^m(v_1),\ldots,G^m(v_n)$ and for $i\in\{1,\ldots,n\}, s\in \{1,\ldots,m\}$ each vertex  $v_{i,s}$ is adjacent to at most one vertex from graphs  $G(C_1), \ldots,G(C_m)$.   Thus,  $G\in \cs_5$.
We will prove that $G$ has an acyclic $(\cs_3,\cs_3)$-colouring if and only if  $\cc\in$ 3NN-SAT.

Assume first that $\cc\in $ 3NN-SAT. We construct the colouring $c$ as follows. If the variable $v_i$ was assigned {\it true}, then all  vertices $v_{i,s}$ are coloured with 1, otherwise we colour them with 2. We will prove that this colouring can be extended to an acyclic $(\cs_3,\cs_3)$-colouring of $G$. Observation \ref{g_vi} yields that this colouring can be extended to an acyclic $(\cs_3,\cs_3)$-colouring of each graph $G^m(v_i)$. Furthermore, all vertices $v_{i,s}$ are 3-saturated, hence a vertex $u_{j,k}$ is coloured with 1, if its corresponding variable in the clause $C_j$ was assigned {\it false}, and with 2 otherwise. Since $C\in$ 3NN-SAT, we have the partition of each set  $\{u_{j,1},u_{j,2},u_{j,3}\}$ into two nonempty parts. Observation  \ref{c_j} yields it can be extended to an acyclic $(\cs_3,\cs_3)$-colouring of each graph  $G(C_j)$. We claim that  such a colouring is an acyclic $(\cs_3,\cs_3)$-colouring of $G$. Indeed, if there is an alternating cycle, then such a cycle passes through  
$v_{i,s}$ and $v_{i,s+1}$, but Observation \ref{sciezkaF} implies there is no alternating path joining them. Thus, the obtained colouring is an acyclic $(\cs_3,\cs_3)$-colouring of $G$.

Now let $c$ be an acyclic $(\cs_3,\cs_3)$-colour\-ing of $G$. Observation
\ref{g_vi} yields for $s=1,\ldots,m$, the vertices $v_{i,s}$  in each graph $G(v_i)$ all have the same colour and are 3-saturated. Thus,  $u_{j,k}$ from the graph
$G(C_j)$ must have a colour distinct from the colour of its neighbour $v_{i,j}$.  Hence for $i=1,\ldots,n$
vertices $u_{j,k}$,
corresponding in the graphs $G(C_j)$ to the variable $c_{j,k}=v_i$ all have the same colour.
Furthermore, in $G(C_j)$, the vertices $u_{j,1},$ $u_{j,2},$ $u_{j,3}$ do not all have the same colour, what follows from Observation \ref{c_jkolor}. Thus, the following assignment $c_{j,k}:=true$ if
$u_{j,k}$ has colour 1, $c_{j,k}:=false$ for otherwise, shows  $\cc\in$ 3NN-SAT.
\end{proof}

\section{Results for graphs with  fixed maximum degree}

Kostochka and Stocker \cite{kost11} proved that there exists a linear-time algorithm computing an acyclic colouring of any graph with maximum degree $d$ with at most $\left\lfloor \frac{(d+1)^2}{4}\right\rfloor+1$ colours. Using a similar method, we prove analogous results for acyclic improper colourings of graphs with maximum degree $d$.
We need a definition first. Recall that  for a given partial $k$-colouring of a graph $G$, a vertex is  rainbow, if all its coloured neighbours  have distinct colours.   A partial $k$-colouring of  $G$ is called {\it rainbow} if  every uncoloured vertex is rainbow. 

\begin{thm}\label{thm_gen1}
For every fixed $d,\;d\ge 4$ there exists a linear $($in $n)$ algorithm finding an acyclic $(d-1)$-improper  colouring for any $n$-vertex graph $G$ with maximum degree at most $d$ using $\left\lfloor \frac{d^2}{4}\right\rfloor +1$ colours.
\end{thm}

\begin{proof}
Let $G$ be a graph with maximum degree $d$ and $c$ be its partial colouring. By $C$ we  denote the set of coloured vertices of $G$. The algorithm proceeds as follows: we choose an uncoloured vertex $v$ with the most number of coloured neighbours, then we greedily colour $v$ with colour $c_i$ such that:

\noindent 1) if $v$ has exactly one coloured neighbour, then 
$c_i\notin c(N(v))\cup c(\bigcup_{x\in N(v)\setminus C}N(x))$;

\noindent 2) if $v$ has more than one  coloured neighbour, then 
$c_i\notin  c(\bigcup_{x\in N(v)\setminus C}N(x))$.

First,  we claim that we can always  find such a colour $c_i$ in $\{1,\ldots ,\left\lfloor \frac{d^2}{4}\right\rfloor +1\}$. Suppose that $v$ has exactly one coloured neighbour. Since $v$ is an uncoloured vertex with the most number of coloured neighbours, each uncoloured neighbour of $v$ has at most one coloured neighbour. Thus,  $|c(N(v))\cup c(\bigcup_{x\in N(v)\setminus C}N(x))|\le d\le \left\lfloor \frac{d^2}{4}\right\rfloor $. Suppose that $v$ has exactly $k$ coloured neighbours. It clearly follows
$|c(\bigcup_{x\in N(v)\setminus C}N(x))|\le (d-k)k\le \left\lfloor \frac{d}{2}\right\rfloor\left\lceil \frac{d}{2}\right\rceil=\left\lfloor \frac{d^2}{4}\right\rfloor$.

Now we show that we eventually obtain an acyclic $(d-1)$-improper  colouring. After each step the partial colouring is rainbow. Thus, if we colour the vertex $v$, then we  do not create any alternating cycle. Finally, it suffices to prove that after each step any coloured vertex has at least one   neighbour coloured with a colour other than its own (if $G$ has more than one coloured vertex). Indeed, if an uncoloured vertex $v$ has exactly one coloured neighbour, then according to our algorithm we colour $v$ with a colour that is not in $C_v$. Suppose now that the uncoloured vertex $v$ has more than one coloured neighbour and till now  after each step any coloured vertex has at least one   neighbour with a colour distinct from its own. The vertex $v$ is rainbow, hence if we  colour $v$, then $v$  will still have this property. Each coloured neighbour of $v$ had this property, hence  colouring $v$ will not destroy this property.  

For the runtime analysis, we propose the following detailed algorithm for an acyclic $(d-1)$-improper  $k$-colouring of a graph $G\in\cs_d$, with $k=\left\lfloor \frac{d^2}{4}\right\rfloor +1$.
The graph $G$ is represented by the lists of incidences. For each vertex $v$ we add to its list of incidences two additional values: $u(v)$ which stores the number of coloured neighbours of $v$, $c(v)$ which stores the colour of $v$. If $v$ is uncoloured, then we have $c(v)=0$. We also maintain $d+1$ lists $A_0, A_1,\ldots, A_d$ to store  vertices and we put a vertex $v$ to $A_j$ if $u(v)=j$.  Initially, we put $c(v)=u(v)=0$ for each vertex $v$, $A_0$ contains all vertices, and  $i=0$, $n=|V(G)|$.
\begin{enumerate}[(1)]
\item {\bf while} $i<n$ {\bf do}
\vspace{-10pt}
\item \tabu choose the largest index $j$ such that $A_j$ is nonempty;
\vspace{-10pt}
\item \tabu choose $v$ to be the first vertex on $A_j$;
\vspace{-10pt}
\item \tabu {\bf if} $u(v)=j$ {\bf then}
\vspace{-10pt}
\item \tabu \{
\vspace{-10pt}
\item \tabu\tabu {\bf if} $u(v)=1$ {\bf then} choose the first colour $a\not\in c(N(v))\cup c(\bigcup_{y\in N(v):c(y)=0}N(y))$
\vspace{-10pt}
\item \tabu\tabu {\bf else} choose the first colour $a\notin  c(\bigcup_{y\in N(v):c(y)=0}N(y))$;
\vspace{-10pt}
\item \tabu\tabu $c(v):=a$; \hfill /* \textit{colour $v$ with $a$} */
\vspace{-10pt}
\item \tabu\tabu delete $v$ from $A_j$;
\vspace{-10pt}
\item \tabu\tabu {\bf for} each neighbour $y$ of $v$ {\bf do}
\vspace{-10pt}
\item \tabu\tabu\tabu {\bf if} $c(y)=0$ {\bf then} \{$u(y)++$; add $y$ to $A_{u(y)}$;\}
\vspace{-10pt}
\item \tabu\tabu $i++$;
\vspace{-10pt}
\item \tabu \}
\vspace{-10pt}
\item \tabu {\bf else} delete $v$ from $A_j$;\hfill /* \textit{$v$ is already coloured} */
\end{enumerate}

\noi Let us compute the running time of the algorithm. Observe first that the {\bf while} loop iterates at most $(d+1)n$ times. Steps (2) and (3) take $\co(d)$ time.  Choosing an admissible colour for a vertex $v$ (in step (6) or (7)) and colour $v$ takes $\co(d^2)$ time, since we have to check the colours of the vertices which are at distance at most two from $v$. Updating the list $A_j$ takes a constant amount of time. The {\bf for} loop in step (10) iterates at most $d$ times, each iteration takes a constant amount of time. The last step also takes a constant amount of time. Hence, for each iteration of the {\bf while} loop we need $\co(d^2)$ time. Thus for fixed $d$, the running time of the algorithm is $\co(n)$. 
\end{proof}

\begin{thm}
Let $d$ and $t$ be fixed such that $2\le t\le \left\lfloor \frac{d}{2}\right\rfloor-1$. There exists a linear $($in $n)$ algorithm finding an acyclic $(d-t)$-improper colouring for any $n$-vertex graph $G$ with maximum degree at most $d$ using $\left\lfloor \frac{d^2}{4}\right\rfloor+t +1$ colours.

\end{thm}

\begin{proof}
Similarly, as in the previous proof, the algorithm will colour vertices of $G$ such that after each step we obtain a rainbow partial colouring $c$ of $G$ with some additional restrictions. In each step of the algorithm we choose an uncoloured vertex $v$ with the most number of coloured neighbours. Let $v$ be the vertex with exactly $k$ coloured neighbours. Observe that  there is an ordering $v_1,v_2,\ldots,v_k$ of coloured neighbours of  $v$ such that the vertex $v_i\;(i=2,\ldots,k)$ had at least  $i-1$ coloured neighbours when it was coloured. Indeed, if we order neighbours of $v$ such that $v_i$ is before $v_j$ if $v_i$ was coloured before $v_j$, then we obtain such an ordering. The vertex $v_i$ had at least $i-1$ coloured neighbours when it was coloured, otherwise $v$ should be coloured before $v_i$. Hence each vertex $v_i\;(i=2,\ldots,k)$ has at least $i-1$ neighbours coloured with distinct colours. We  colour $v$ with colour $c_i$ such that:

\noindent 1) if $k\le t$, then $c_i\notin c(\{v_1,v_2,\ldots,v_k\})\cup c(\bigcup_{x\in N(v)\setminus C}N(x))$;

\noindent 2) if $k> t$, then $c_i\notin c(\{v_1,v_2,\ldots,v_t\})\cup c(\bigcup_{x\in N(v)\setminus C}N(x))$.

First,  we claim that we can always  find such a colour $c_i$ in $\{1,\ldots,\left\lfloor \frac{d^2}{4}\right\rfloor +t+1\}$. Since the vertex $v$ has $k$ coloured neighbours and it is the vertex with the most number of coloured neighbours, we have  $|c(\bigcup_{x\in N(v)\setminus C}N(x))|\le (d-k)k\le \left\lfloor \frac{d}{2}\right\rfloor\left\lceil \frac{d}{2}\right\rceil=\left\lfloor \frac{d^2}{4}\right\rfloor $.

Now, we show that after each step the obtained partial colouring has no alternating cycle and each coloured vertex $v$ either has at least $t$ coloured neighbours which  have colours distinct from its own or all its coloured neighbours  have distinct colours, whenever $v$ has less than $t+1$ coloured neighbours. If we colour the first vertex, then it obviously holds. Suppose that till now after each step this property holds and now we colour the vertex $v$ which has $k$ coloured neighbours.
Since $v$ is rainbow, colouring $v$ does not create any alternating cycle.  If $k\le t$, then we colour $v$ with a colour distinct from the colours of its neighbours. Thus, $v$ will have all neighbours coloured with distinct colours and each neighbour $u$ of $v$  will either have  at least $t$ neighbours coloured with  colours distinct from its own or all coloured neighbours of $u$ have distinct colours, whenever $u$ has less than $t+1$ coloured neighbours. Suppose now that $k>t$. We colour $v$ with the colour distinct from the colours of $v_1,v_2,\ldots,v_t$. Thus, $v$ has at least $t$ neighbours coloured with colours distinct from its own colour. Before we colour $v$, any vertex among $v_{t+1},\ldots,v_k$ had at least $t$ coloured neighbours, so each of $v_{t+1},\ldots, v_k$ had at least $t$ neighbours coloured with a colour distinct from its own. So after colouring $v$ they still have this property.  Since we colour $v$ with a colour distinct from the colours of $v_1,v_2,\ldots, v_t$, after colouring $v$ the  neighbours of  $v_1,v_2,\ldots, v_t$ have this property. Hence the algorithm creates a $(d-t)$-improper colouring. For the running time, it is enough to observe that an algorithm, similar to that  given in the proof of Theorem \ref{thm_gen1}, will work.
\end{proof}

It is an easy observation that if a graph $G$ admits an acyclic $t$-improper colouring, then $G$ also has an acyclic $p$-improper colouring, for any $p\geq t$. From this fact and since an acyclic colouring can be always treated as an acyclic $s$-improper colouring, with $s=0$, the next result follows directly from the aforementioned theorem of Kostochka and Stocker \cite{kost11}.

\begin{cor}
Let $d$ and $t$ be fixed such that $t\le \left\lceil \frac{d}{2}\right\rceil$. There exists a linear $($in $n)$ algorithm finding an acyclic $t$-improper colouring for any $n$-vertex graph $G$ with maximum degree at most $d$ using $\left\lfloor \frac{(d+1)^2}{4}\right\rfloor +1$ colours.
\end{cor}

\section{Concluding remarks}

In the paper we consider acyclic improper colourings of graphs with bounded degree. In particular, we give a linear-time algorithm for an acyclic $t$-improper colouring of any graph with maximum degree at most $d$, provided that the number of colours used is large enough with respect to $d$. Fixing the maximum degree to five, we obtain results which are more exact. Namely, we show that every such graph has an acyclic colouring with five colours in which each colour class induces an acyclic graph (Theorem \ref{acyclic_5}) and we further improve this result (in Theorem \ref{thm_s5}) as follows: Every graph with maximum degree at most five admits an acyclic colouring with five colours such that each colour class induces an acyclic graph with maximum degree at most four. 

It might be  interesting,  how  the above mentioned result can be further improved. One possible way is to put stronger condition on the colour classes. We post the following problem:

\begin{op}
Let $G$ be a graph with maximum degree at most five. For which properties $\cal P$ graph $G$ admits an acyclic $({\cal P})^{(5)}$-colouring? 
\end{op}

We prove that the property  ${\cal D}_1\cap {\cal S}_4$ (of acyclic graphs with maximum degree at most four) can be taken as $\cal P$. Finding the smallest  such property is a challenging question. If we additionally assume that the property $\cal P$ has to be hereditary, then the above problem coincides with the problem of finding the acyclic reducible bounds for ${\cal S}_5$ (see, e.g., \cite{bofi06}, \cite{bofi09} for the definition and some results  concerning acyclic reducible bounds).

On the other hand, one can think about reducing the number of colours used. Since the complete graph on five vertices is, as far as we know, the only one graph that actually requires 5 colours, maybe it can be possible to use less colours, while excluding $K_5$.

\section{Acknowledgements}
\noindent The research of the first author was supported by the Minister of Science and Higher Education of Poland (grant no. JP2010 009070).


\end{document}